\documentclass[journal, onecolumn, draftclsnofoot]{IEEEtran}
\ifCLASSINFOpdf
\else
\fi

\usepackage{graphicx} 
\usepackage{amsmath}
\usepackage{amssymb}
\usepackage{pgf}
\usepackage{tikz}
\usepackage{enumerate}
\usepackage{url}

\newtheorem{theorem}{Theorem}[section]
\newtheorem{lemma}[theorem]{Lemma}

\newtheorem{problem}[theorem]{Problem}

\newtheorem{remark}[theorem]{Remark}
\newtheorem{ex}[theorem]{Example}

\begin{document}

\title{Control Design for Markov Chains under \\ Safety Constraints:
A Convex Approach}

\author{  Eduardo~Arvelo and Nuno C. Martins
\thanks{This work was supported by NSF Grant CNS  0931878,  AFOSR Grant FA95501110182, ONR UMD-AppEl Center and the Multiscale Systems Center, one of six research centers funded under the Focus Center Research Program.}
\thanks{E. Arvelo is with the Electrical and Computer Engineering Department, University of Maryland, College Park, MD 20742, USA. email: earvelo@umd.edu.}% <-this % stops a space
\thanks{N. C. Martins is with the Faculty of Electrical and Computer Engineering Department and the Institute for Systems Research, University of Maryland, College Park, MD 20742, USA. email: nmartins@umd.edu.}}% <-this 

\maketitle

\begin{abstract}
This paper focuses on the design of time-invariant memoryless control policies for fully observed controlled Markov chains, with a finite state space. Safety constraints are imposed through a pre-selected set of forbidden states. A state is qualified as safe if it is not a forbidden state and the probability of it transitioning to a forbidden state is zero. The main objective is to obtain control policies whose closed loop generates the maximal set of safe recurrent states, which may include multiple recurrent classes. A design method is proposed that relies on a finitely parametrized convex program inspired on entropy maximization principles. A numerical example is provided and the adoption of additional constraints is discussed.
\end{abstract}

%\IEEEpeerreviewmaketitle

%%%%%%%%%%
%INTRODUCTION%
%%%%%%%%%%
\section{Introduction}

The formalism of controlled Markov chains is widely used to describe the behavior of systems whose state transitions probabilistically among different configurations over time. Control variables act by dictating the state transition probabilities, subject to constraints that are specified by the model. Existing work has addressed the design of controllers that optimize a wide variety of costs that depend linearly on the parameters that characterize the probabilistic behavior of the system. The two most commonly used tools are linear and dynamic programming. For an extensive survey, see \cite{Arapostathis:1993vv} and the references therein.

We focus on the design of time-invariant memoryless policies for fully observable controlled Markov chains with finite state and control spaces, represented as $\mathbb{X}$ and $\mathbb{U}$, respectively.
Given a pre-selected set $\mathbb{F}$ of forbidden states of $\mathbb{X}$, a state is qualified as $\mathbb{F}$-safe if it is not in $\mathbb{F}$ and the probability of it transitioning to an element of $\mathbb{F}$ is zero.
Here, forbidden states may represent unwanted configurations. We address a problem on control design subject to \underline{safety constraints} that consists on finding a control policy that leads to the maximal set of $\mathbb{F}$-safe recurrent states $\mathbb{X}_{\mathbb{F}}^R$. This problem is relevant when persistent state visitation is desirable for the largest number of states without violating the safety constraint, such as in the context of persistent surveillance. 

We show in Section~\ref{sec:recurrent} that the maximal set of $\mathbb{F}$-safe recurrent states $\mathbb{X}_{\mathbb{F}}^R$ is well defined and achievable by suitable control policies. 
As we discuss in Remark~\ref{213317j2},  $\mathbb{X}_{\mathbb{F}}^R$ may contain multiple recurrent classes, but does not intersect the set of forbidden states $\mathbb{F}$. %(see Fig.~\ref{fig:sets}).
 
\subsection{Comparison with existing work}
Safety-constrained controlled Markov chains have been studied in a series of papers by Arapostathis et al., where the state probability distribution is restricted to be bounded above and below by safety vectors at all times. In \cite{Arapostathis:2003gs}, \cite{Arapostathis:2005bk} and \cite{Hsu:2010th}, the authors propose algorithms to find the set of distributions whose evolution under a given control policy respect the safety constraint. In \cite{Wu:2004ik}, an augmented Markov chain is used to find the the maximal set of probability distributions whose evolution respect the safety constraint over all admissible non-stationary control policies.

Here we are not concerned with the maximization of a given performance objective, but rather in systematically characterizing the maximal set of $\mathbb{F}$-safe recurrent states and its corresponding control policies. The main contribution of this paper is to solve this problem via finitely parametrized convex programs. Our approach is rooted on entropy maximization principles, and the proposed solution can be easily implemented using standard convex optimization tools, such as the ones described in \cite{cvx}.

\subsection{Paper organization}
The remainder of this paper is organized as follows. Section \ref{sec:notation} provides notation, basic definitions and the problem statement. The convex program that generates the maximal set of $\mathbb{F}$-safe recurrent states is presented in Section \ref{sec:recurrent} along with a numerical example. Further considerations are given in Section \ref{sec:further}, while conclusions are discussed in Section~\ref{sec:conclusion}.

%%%%%%%%%%%
%PRELIMINARIES%
%%%%%%%%%%%
\section{Preliminaries and Problem Statement}\label{sec:notation}

The following notation is used throughout the paper: 
\vspace{2mm}

\begin{tabular}{l l}
	$\mathbb{X}$ & state space of the Markov chain \\
	$\mathbb{U}$ & set of control actions \\ 
	$n$ & cardinality of $\mathbb{X}$ \\
	$m$ & cardinality of $\mathbb{U}$ \\
	$X_k$ & state of the Markov chain at time $k$\\
	$U_k$ & control action at time $k$\\
	$\mathbb{P}_\mathbb{X}$ & set of all pmfs with support in $\mathbb{X}$\\
	$\mathbb{P}_\mathbb{U}$ &  set of all pmfs with support in $\mathbb{U}$ \\
	$\mathbb{P}_\mathbb{XU}$ &  set of all joint pmfs with support in $\mathbb{X} \times \mathbb{U}$\\
	$\mathbb{S}_{f}$ & support of a pmf $f$
\end{tabular}

\vspace{3mm}
The recursion of the controlled Markov chain is given by the (conditional) probability mass function of $X_{k+1}$ given the previous state $X_k$ and control action $U_k$, and is denoted as:
\begin{align*}
	\mathcal{Q}(x^+, x, u) \overset{def}{=} P(X_{k+1} = x^+| X_k = x, U_k = u), \qquad x^+, x \in \mathbb{X}, ~u \in \mathbb{U}, 
\end{align*}
We denote any memoryless time-invariant control policy by
\begin{align*}
	\mathcal{K}(u, x) \overset{def}{=}P(U_k = u| X_k = x), \qquad u \in \mathbb{U}, x \in \mathbb{X}
\end{align*}
where $\sum_{u \in\mathbb{U}}\mathcal{K}(u,x)=1$ for all $x$ in $\mathbb{X}$. The set of all such policies is denoted as $\mathbb{K}$.
\paragraph*{Assumption} Throughout the paper, we assume that the controlled Markov chain $\mathcal{Q}$ is given. Hence, all quantities and sets that depend on the closed loop behavior will be indexed only by the underlying control policy $\mathcal{K}$.

Given a control policy $\mathcal{K}$, the conditional state transition probability of the closed loop is
represented as:
\begin{align}\label{eqn:onestep}
	P_{\mathcal{K}}(X_{k+1} & = x^+| X_k = x) \overset{def}{=} \sum_{u\in\mathbb{U}}\mathcal{Q}(x^+,x,u)\mathcal{K}(u,x), \qquad x^+,x \in \mathbb{X}
\end{align}
We define the \emph{set of recurrent states} $\mathbb{X}^R_{\mathcal{K}}$ and the \emph{set of $\mathbb{F}$-safe recurrent states} $\mathbb{X}^R_{\mathcal{K},\mathbb{F}}$ under a control policy $\mathcal{K}$ to be:
\begin{align*}
	&\mathbb{X}^R_{\mathcal{K}}\! \overset{def}{=}\!\!
	\Big\{x \in \mathbb{X}: P_{\mathcal{K}} (X_{k} \!=\! x \text{ for some } k>0| X_0 \!=\! x)\!=\!1 \Big\} \\
	&\mathbb{X}^R_{\mathcal{K},\mathbb{F}}\! \overset{def}{=}\!\! \Big\{  x \in \mathbb{X}^R_{\mathcal{K}} \text{ : } P_{\mathcal{K}} (X_{k+1}\!=\!x^+ | X_k \!=\!x) \!=\! 0 \text{, } x^+ \in \mathbb{F} \Big\}
\end{align*}
The maximal set of $\mathbb{F}$-safe recurrent states is defined as:
\begin{align*}
	\mathbb{X}^R_{\mathbb{F}} \overset{def}{=} \bigcup_{\mathcal{K} \in \mathbb{K} } \mathbb{X}^R_{\mathcal{K},\mathbb{F}}
\end{align*}
The problem we address in this paper is defined below:
%
%PROBLEM 1 (RECURRENT)
%\hlinee 
\begin{problem}\label{prob1}
Given  $\mathbb{F}$, determine $\mathbb{X}^R_{\mathbb{F}}$ and a control policy $\mathcal{K}^{*}_R$ such that $\mathbb{X}^R_{\mathcal{K}^{*}_R} = \mathbb{X}^R_{\mathbb{F}}$. 
\end{problem} A solution to Problem~\ref{prob1} is provided in Section~\ref{sec:recurrent}, where we also show the existence of $\mathcal{K}^{*}_R$.
%\hlinee
%
\begin{remark} \label{213317j2} The following is a list of important observations on Problem~\ref{prob1}:
\begin{itemize}
	\item  The set $\mathbb{X}_{\mathbb{F}}^R$ may contain more than one recurrent class and it will exclude any recurrent class that intersects $\mathbb{F}$.
	\item There is no $\mathcal{K}$ such that the states in $\mathbb{X} \diagdown \mathbb{X}_{\mathbb{F}}^R$  can be
	          $\mathbb{F}$-safe and recurrent.
	\item If the closed loop Markov chain is initialized in $\mathbb{X}_{\mathbb{F}}^R$ then the probability that it will ever visit a state in $\mathbb{F}$ is zero.
\end{itemize}
\end{remark}
%	
%	
%
%%%%%%%%%%%%%%%%%%%%%%
%MAIN SECTION: RECURRENT STATES%
%%%%%%%%%%%%%%%%%%%%%%
\section{Maximal $\mathbb{F}$-Safe Set of Recurrent States}\label{sec:recurrent}
We propose a convex program to solve problem \ref{prob1}. Consider now the following convex optimization program:
\begin{align}
	&f^*_{XU} =  \arg \max_{f_{XU} \in \mathbb{P}_\mathbb{XU}} \mathcal{H}(f_{XU})\label{eqn:opt1a}\\
	&\text{subject to: } \nonumber\\
	&\sum_{u^+ \in\mathbb{U}}f_{XU}(x^+,u^+) = \sum_{x \in \mathbb{X}, u\in\mathbb{U}}\mathcal{Q}(x^+, x, u)f_{XU}(x,u),~~ x^+ \in \mathbb{X}\label{eqn:opt1b}\\
	&\sum_{u \in \mathbb{U}} f_{XU}(x,u) = 0, ~~ x \in \mathbb{F}\label{eqn:opt1c}
\end{align}

where $\mathcal{H}: \mathbb{P}_{\mathbb{X}\mathbb{U}} \rightarrow \Re_{\geq 0}$ is the entropy of $f_{XU}$, and is given by
\begin{align*}
	\mathcal{H}(f_{XU}) = -\sum_{u \in \mathbb{U}}\sum_{x \in \mathbb{X}}f_{XU}(x,u)\ln(f_{XU}(x,u))
\end{align*} where we adopt the standard convention that $0 \ln ( 0 ) = 0$.

The following Theorem provides a solution to Problem~\ref{prob1}.
%
%THEOREM MAX RECURRENT.
\begin{theorem}\label{thm1}
Let $\mathbb{F}$ be given, and assume that (\ref{eqn:opt1a})-(\ref{eqn:opt1c}) is feasible and that $f^*_{XU}$ is the optimal solution. In addition, adopt the marginal pmf $f^*_X(x) = \sum_{u \in \mathbb{U}} f^*_{XU}(x,u)$
and consider that $\mathcal{G}: \mathbb{U} \times \mathbb{X} \rightarrow [0,1]$ is any 
function satisfying $\sum_{u \in \mathbb{U}} \mathcal{G} (u,x)=1$ for all $x$ in $\mathbb{X}$.
The following holds:
\begin{enumerate}[(a)]
	\item $\mathbb{X}^R_{\mathcal{K}, \mathbb{F}} \subseteq \mathbb{S}_{f^*_X}$ for all $\mathcal{K}$ in $\mathbb{K}$.
	\item $\mathbb{X}^R_{\mathbb{F}} = \mathbb{S}_{f^*_X}$
	\item $\mathbb{X}^R_{\mathcal{K}^*_R} = \mathbb{X}^R_{\mathbb{F}}$ for $\mathcal{K}^*_R$ given by:
	\begin{align}\label{eqn:control_opt}
		\!\!\!\!\!\!\!\! \mathcal{K}^*_R(u,x) =
		\begin{cases}
		 	\frac{f^*_{XU}(x,u)}{f^*_X(x)}, & \!x \in \mathbb{S}_{f^*_X} \\
			\mathcal{G}(u,x),\! &\text{otherwise}
		\end{cases}, \quad (u,x) \in \mathbb{U}\times\mathbb{X}
	\end{align}
\end{enumerate}
where we use $\mathbb{S}_{f^*_X} = \{x\in \mathbb{X} \text{ : }  f^*_X(x)>0\}$.
\end{theorem}
The proof of Theorem \ref{thm1} is given at the end of this section.
%
%%%%%%%%
%EXAMPLE1%
%%%%%%%%
\begin{ex}~\label{sec:ex1} {\bf (Computation of the maximal set of $\mathbb{F}$-safe recurrent states)}
 Suppose $\mathbb{X} = \{1, ... ,8\}$ and $\mathbb{U}=\{1, 2\}$, and consider the controlled Markov chain whose probability transition matrix is  $Q^u$, where $Q^u_{ij} \overset{def}{=} \mathcal{Q}(i, j, u)$:

{\begin{align*}
	\begin{array}{lll}
	Q^1\!\!=\!\!\!
	\begin{bmatrix}
	0 & .5 & 0 & .5 & 0 & 0 & 0 & 0\\
	.3 & .7 & 0 & 0 & 0 & 0 & 0 & 0\\
	.2 & 0 & .5 & 0 & 0 & 0 & 0 & .3\\
	0 & 0 & 0 & 0 & 0 & 0 & 1 & 0\\
	0 & 0 & 0 & 0 & 0 & 1 & 0 & 0\\
	0 & 0 & 0 & .1 & .3 & 0 & 0 & .6\\
	0 & 0 & 0 & 0 & 0 & 0 & 1 & 0\\
	0 & 0 & 0 & 0 & 0 & 0 & 1 & 0\\
	 \end{bmatrix}\!\!\!,
	Q^2 \! \!=\! \!\!
	\begin{bmatrix}
	0 & 0 & 1 & 0 & 0 & 0 & 0 & 0\\
	0 & 0 & 1 & 0 & 0 & 0 & 0 & 0\\
	.2 & .8 & 0 & 0 & 0 & 0 & 0 & 0\\
	1 & 0 & 0 & 0 & 0 & 0 & 0 & 0\\
	0 & .2 & 0 & 0 & .2& 0 & .6 & 0\\
	0 & 0 & .2 & .8 & 0 & 0 & 0 & 0\\
	0 & 0 & 0 & 0 & 0 & 0 & 0 & 1\\
	0 & 0 & 0 & 0 & 0 & .1 & 0 & .9\\
	 \end{bmatrix}
	\end{array}
\end{align*}
}

\begin{figure}[t]
\centering
\definecolor{cffffff}{RGB}{255,255,255}

\input{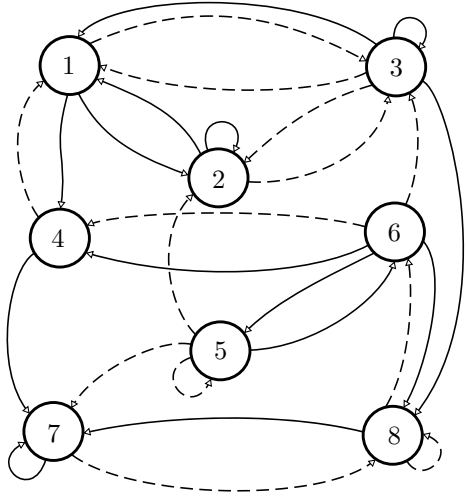}

\caption{An example of a controlled Markov chain with eight states and two control actions. The directed arrows indicate state transitions with positive probability. The solid and dashed lines represent control actions $1$ and $2$, respectively. }
\label{fig:example2}
\end{figure}

The interconnection diagram of this controlled Markov chain can be seen in Fig. \ref{fig:example2}. Suppose that state $4$ is undesirable and should never be visited (in other words, $\mathbb{F} = \{4\}$). By solving (\ref{eqn:opt1a})-(\ref{eqn:opt1c}), we find that:
\begin{align*}
	F^* = \begin{bmatrix}
	0&   .16 &   0~ &   0~ &  0~  &  0~   & .15&  .15\\
	.08  &  .11 &  .2~   & 0~ & 0~ & 0~ & .15  & 0\\
	\end{bmatrix},
\end{align*}
where $F^*_{ij} \overset{def}{=} f^*_{XU}(i, j)$. This implies by Theorem \ref{thm1} that:
\begin{align*}
	\mathbb{X}^R_{\mathbb{F}} = \{1, 2, 3, 7, 8\},
\end{align*}
and the associated control policy is given by the matrix:
\begin{align*}
	K^*=
	\begin{bmatrix}
	0~ &  .58~ & 0~   & G_{1,4}~ &  G_{1,5} ~   & G_{1,6}~   & .5 & 1 \\
	1 ~ & .42~ & 1~   & G_{2,4}~ &  G_{2,5} ~  & G_{2,6} ~ & .5 &  0\\
	\end{bmatrix},
\end{align*}
where $K_{ij}^*\overset{def}{=}\mathcal{K}^*_R(i,j)$; and $G \in [0,1]^{2 \times 8}$ is any matrix whose the columns sum up to 1.

It is important to highlight some interesting points that would otherwise not be clear if we were to consider a large system. Note that state $6$ is not in $\mathbb{X}_{\mathbb{F}}^R$ because regardless of which control action is chosen the probability of transitioning to $\mathbb{F}$ is always a positive. State $5$ cannot be made recurrent even though it is a safe state. Furthermore, when the chain visits states $1$ and $8$, one of the two available control actions cannot be chosen since that choice leads to a positive probability of reaching $\mathbb{F}$. In this scenario there are two safe recurrent classes: \{1,2,3\} and \{7,8\}. Note that the control action $1$ cannot be chosen when the chain visits state $3$ because that choice makes the states $1, 2, \text{and }  3$ transient. 
\end{ex}

\begin{remark} Consider a control policy $\mathcal{K}$ for which $\mathcal{K}(x,u) > 0$ if and only if  $\mathcal{K}^*_R(x,u) > 0$, where $\mathcal{K}^*_R$ is an optimal solution given by (\ref{eqn:control_opt}). It holds that $\mathbb{X}_{\mathcal{K}}^R =\mathbb{X}_{\mathcal{K}^*_R}^R = \mathbb{X}_\mathbb{F}^R$. \end{remark}

\begin{remark}\label{rem2}
Consider a control policy $\mathcal{K}$ for which $\mathbb{X}_{\mathcal{K}}^R = \mathbb{X}_\mathbb{F}^R$. It holds that $\mathcal{K}(x,u) = 0$ if $\mathcal{K}^*_R(x,u) = 0$.
\end{remark}

\subsection{Proof of Theorem \ref{thm1}}
To facilitate the proof we first establish the following Lemma:
%LEMMA ENTROPY
\begin{lemma}\label{lemma_entropy}
Let  $\mathbb{Y}$ be a finite set and $\mathbb{W}$ be a convex subset of  $\mathbb{P}_{\mathbb{Y}}$, the set of all pmfs with support in $\mathbb{Y}$. Consider the following problem:
\begin{align*}
	&f^* =  \arg \max_{f \in \mathbb{W}} \mathcal{H}(f)\\
\end{align*}
where $\mathcal{H}(f)$ is the entropy of $f$ in $\mathbb{P}_{\mathbb{Y}}$ and is given by $ \mathcal{H}(f) = -\sum_{y \in \mathbb{Y}}f(y)\ln(f(y))$, where we adopt the convention that $0 \ln (0) = 0$.
The following holds:
\begin{align*}
	 \mathbb{S}_{f} \subseteq \mathbb{S}_{f^*}, ~ f \in \mathbb{W}
\end{align*}
	where $\mathbb{S}_{f} = \{y\in \mathbb{Y}| f(y)>0\}$, $f \in \mathbb{P}_{\mathbb{Y}}$.
\end{lemma}

%PROOF
\begin{proof} Select an arbitrary $f$ in $\mathbb{W}$ and define $f_\lambda \overset{def}{=} \lambda f^*+ (1-\lambda)f$ for $0\leq \lambda \leq 1$. From the convexity of $\mathbb{W}$, we conclude that $f_\lambda$ is in $\mathbb{W}$ for all $\lambda$ in $[0,1]$. Since $f^*$ has maximal entropy, it must be that there exists a $\bar{\lambda}$ in $ [0,1)$ such that
\begin{align}\label{eqn:optcond}
	\frac{d}{d\lambda}\mathcal{H}(f_\lambda)\geq 0, ~ \lambda \in (\bar{\lambda},1).
\end{align}
\paragraph*{Proof by contradiction:} Suppose that $\mathbb{S}_f \nsubseteq \mathbb{S}_{f^*}$ and hence that there exists a $y'$ in $\mathbb{Y}$ such that $f(y')>0$ and $f^*(y')=0$. We have that
\begin{align*}
	\frac{d}{d\lambda} \bigg ( f_{\lambda}(y') \ln (f_{\lambda}(y')) \bigg )=  -f(y') \big ( \ln(f_\lambda(y)) +1 \big )
\end{align*}
goes to $\infty$, as $\lambda$ approaches $1$, since $\lim_{\lambda \rightarrow 1} f_\lambda(y') = 0$. This implies that there exists a $\tilde{\lambda}$ in $ [0,1)$ such that 
\begin{align*}
	\frac{d}{d\lambda}\mathcal{H}(f_\lambda) < 0, ~  \lambda \in  (\tilde{\lambda},1 ),
\end{align*}
which contradicts (\ref{eqn:optcond}). 
\end{proof}
See (\cite{Csiszar:2004vd}) for an alternative proof that relies on the concept of relative entropy.

%PROOF FOR THEOREM REC
\begin{proof}[Proof of Theorem \ref{thm1}]\
\begin{enumerate}[(a)]
	\item (Proof that \emph{$\mathbb{X}^R_{\mathcal{K}, \mathbb{F}} \subseteq \mathbb{S}_{f^*_X}$} holds for all $\mathcal{K}$ in $\mathbb{K}$.) Select an arbitrary control policy $\mathcal{K}$ in $\mathbb{K}$. There are two possible cases: {\bf i)} When $\mathbb{X}^R_{\mathcal{K}, \mathbb{F}}$ is the empty set, the statement follows trivially. {\bf ii)} If $\mathbb{X}^R_{\mathcal{K}, \mathbb{F}}$ is non-empty then the closed loop must have an invariant pmf $f^{\mathcal{K}}_{XU}$ that satisfies the following:  
		\begin{align}\label{eqn:invcandidate}
			f^{\mathcal{K}}_{XU}&(x^+,u^+) = \mathcal{K}(u^+, x^+)\sum_{x \in \mathbb{X}, u\in\mathbb{U}}\mathcal{Q}(x^+, x, u)f^{\mathcal{K}}_{XU}(x,u),
		\end{align}
		\begin{align}	\label{1016280812}		
			&\sum_{u\in\mathbb{U}}f^{\mathcal{K}}_{XU}(x,u)>0,~ x \in \mathbb{X}^R_{\mathcal{K}, \mathbb{F}}; & \text{(recurrence)} \\ \label{1017280812}
			&\sum_{u\in\mathbb{U}}f^{\mathcal{K}}_{XU}(x,u)=0,~ x \in \mathbb{F}.& \text{($\mathbb{F}$-safety)}
		\end{align} 

		Equation (\ref{eqn:invcandidate}) follows from the fact that $f^{\mathcal{K}}_{XU}$ is an  invariant pmf of the closed loop, while (\ref{1016280812})-(\ref{1017280812}) follow from the definition of $\mathbb{X}^R_{\mathcal{K},\mathbb{F}}$.
		
		Our strategy to conclude the proof (that $\mathbb{X}^R_{\mathcal{K}, \mathbb{F}} \subseteq \mathbb{S}_{f^*_X}$ holds) is to show that $f^{\mathcal{K}}_{XU}$ is a feasible solution of the convex program (\ref{eqn:opt1a})-(\ref{eqn:opt1c}), after which we can use Lemma \ref{lemma_entropy}. In order to show the aforementioned feasibility, note that summing both sides of equation (\ref{eqn:invcandidate}) over the set of control actions yields (\ref{eqn:opt1b}), where we use the fact that $\sum_{u^+\in\mathbb{U}}\mathcal{K}(u^+,x^+)=1$. Moreover, the constraint (\ref{eqn:opt1c}) and the $\mathbb{F}$-safety equality in (\ref{1017280812}) are identical. Therefore, $f^{\mathcal{K}}_{XU}$ is a feasible pmf to (\ref{eqn:opt1a})-(\ref{eqn:opt1c}). By Lemma \ref{lemma_entropy}, it follows that $\mathbb{S}_{f^{\mathcal{K}}_X} \subseteq \mathbb{S}_{f^*_X}$ and, consequently, that $\mathbb{X}^R_{\mathcal{K}, \mathbb{F}} \subseteq \mathbb{S}_{f^*_X}$.

	\item (Proof that \emph{$\mathbb{X}^R_{\mathbb{F}} = \mathbb{S}_{f^*_X}$} holds ) It follows from (a) that $\mathbb{X}^R_{\mathbb{F}} \subseteq \mathbb{S}_{f^*_X}$. To prove that $\mathbb{S}_{f^*_X} \subseteq\mathbb{X}^R_{\mathbb{F}}$, select an optimal policy $\mathcal{K}_R^*$ as in (\ref{eqn:control_opt}), and note that the corresponding closed loop pmf $f^*_{XU}$ is an invariant distribution, leading to:
		\begin{align*}
		f^*_{XU}(x^+,&u^+) = \mathcal{K}_R^*(u^+, x^+)\sum_{x \in \mathbb{X}, u\in\mathbb{U}}\mathcal{Q}(x^+, x, u)f^*_{XU}(x,u).
		\end{align*}
		Consider any element $\hat{x}$ in $\mathbb{S}_{f^*_X}$. Since 
			${\sum_{u\in\mathbb{U}}f^*_{XU}(\hat{x},u)>0}$ holds
		and from the fact that $f^*_{XU}$ is an invariant distribution of the closed loop, we conclude that the following must be satisfied:
		\begin{align*}
			P_{\mathcal{K}_R^*}(X_{k} = \hat{x} \text{ for some } k>0| X_0 = \hat{x})=1.
		\end{align*}
		This means that $\hat{x}$ belongs to $\mathbb{X}^R_{\mathcal{K}_R^*}$. From (\ref{eqn:opt1c}), it is clear that $\hat{x}$ is an $\mathbb{F}$-safe state and, thus, belongs to $\mathbb{X}^R_{\mathcal{K}_R^*,\mathbb{F}}$. Hence, by definition, $\hat{x}$ belongs to  $\mathbb{X}^R_{\mathbb{F}}$. Since the choice of $\hat{x}$ in $\mathbb{S}_{f^*_X}$ was arbitrary, we conclude that $\mathbb{S}_{f^*_X} \subseteq\mathbb{X}^R_{\mathbb{F}}$.

	\item (Proof that  \emph{$\mathbb{X}^R_{\mathcal{K}^*_R} = \mathbb{X}^R_{\mathbb{F}}$} holds) Follows from the proof of (b). 
\end{enumerate}
\end{proof}

\section{Further Considerations}\label{sec:further}

\noindent {\bf Computational complexity reduction. } Consider the following convex program:
\begin{align*}
	&\bar{f}_{XU} =  \arg \max_{\begin{array}{c}
f_{XU} \in \mathbb{P}_\mathbb{XU} \\
\text{s.t. } (\ref{eqn:opt1b}) \text{ and }(\ref{eqn:opt1c})\end{array}}\mathcal{H}\Big(\sum_{u\in\mathbb{U}}f_{XU}(\cdot, u)\Big),\nonumber 
\end{align*}
where the objective function has been modified to be the entropy of the marginal pmf with respect to the state (rather than the joint entropy as in (\ref{eqn:opt1a})). A simple modification of Theorem \ref{thm1} leads to the conclusion that  $\mathbb{S}_{\bar{f}_X} = \mathbb{S}_{f^*_X}$. Therefore, the modified program also provides a solution to Problem \ref{prob1} with the advantage that it requires fewer calls to the entropy function, thus reducing computational complexity. However, the optimal control policy $\bar{\mathcal{K}}_R$ (obtained in an analogous manner as in (\ref{eqn:control_opt})) may differ from $\mathcal{K}^*_R$. The most significant difference is that Remark \ref{rem2} does not apply to $\bar{\mathcal{K}}_R$.\\

\noindent {\bf Additional constraints.} Further convex constraints on $f_{XU}$ may be incorporated in the optimization problem (\ref{eqn:opt1a})-(\ref{eqn:opt1c}) without affecting its tractability. For instance, consider constraints of the following type:
\begin{align}\label{eqn:const}
\sum_{u \in\mathbb{U}, x\in\mathbb{X}}h(u)f_{XU}(x,u) \leq \beta,
\end{align}
where $h$ is an arbitrary function and $\beta$ an arbitrary real number. Let $f^*_{XU}$ be an optimal solution to (\ref{eqn:opt1a})-(\ref{eqn:opt1c}) and (\ref{eqn:const}). If the Markov chain is initialized with  $f^*_{XU}$ (an invariant distribution), the following holds for each $k$: 
\begin{align*}
E[h(U_k)]\leq \beta.
\end{align*}
Moreover, when $\mathbb{X}_{\mathbb{F}}^R$ contains only one aperiodic recurrent class, the following holds for any initial distribution:
\begin{align*}
\lim_{k\rightarrow\infty} E[h(U_k)] \leq \beta.
\end{align*}

\section{Conclusion}\label{sec:conclusion}
This paper addresses the design of full-state feedback policies for controlled Markov chains defined on finite alphabets. The main problem is to design policies that lead to the largest set of recurrent states, for which the probability to transition to a pre-selected set of forbidden states is zero. The paper describes a finitely parametrized convex program that solves the problem via entropy maximization principles.

\nocite{Garg:1999if}
\nocite{HernandezLerma:1996ti}
\nocite{Puterman:1994vo}
\nocite{Wolfe:1962ua}
\nocite{Borkar:1990cz}
\nocite{Fox:1966ut}
\nocite{Altman:1999tz}
\nocite{springerlink:10.1007/978-1-4615-0805-2_11}
\nocite{Bertsekas:1987tx}
\nocite{Hordijk:1979uv}
\nocite{Borkar:1989ty}
\nocite{Kumar:1986ww}

\bibliographystyle{plain}
\bibliography{EA}

\end{document}